\documentclass[12pt,a4paper]{article}
\usepackage{authblk}
\usepackage{lmodern}
\usepackage[T1]{fontenc}
\usepackage[cp1250]{inputenc}
\usepackage{amssymb}
\usepackage{pstricks}
\usepackage{pst-node}
\usepackage[normalem]{ulem}
\usepackage{amsmath}
\usepackage{amsthm}
\usepackage{amsfonts}
\usepackage{tcolorbox}

\newtheorem{theorem}{Theorem}[section]
\newtheorem{lemma}[theorem]{Lemma}

\makeatletter

\newcommand{\mmath}[1]{\relax\ifmmode#1\relax\else$#1$\fi}

\newcommand{\unmath}[1]{\relax\ifmmode\hbox{#1}\else#1\fi}

\let\ifnextchar=\@ifnextchar

\def\getprime#1{\let\oldpr\primes \edef\primes{\oldpr#1}%
\ifnextchar'{\getprime}{\primes$}}

\makeatother

\newcommand{\remove}[1]{\relax}

\newcommand{\defeq}{\stackrel{\text{def}}{=}}

\newcommand{\len}[1]{\ensuremath{\|#1\|}}

\newcommand{\VV}{\ensuremath{{\cal V}}}

\newcommand{\GG}{\ensuremath{{\cal G}}}
\newcommand{\GGt}{\ensuremath{\GG(\theta)}}
\newcommand{\GGte}{\ensuremath{\GG(\theta+\varepsilon)}}
\renewcommand{\SS}{\ensuremath{{\cal S}}}

\newcommand{\QQ}{\ensuremath{{\cal Q}}}

\renewcommand{\P}{\ensuremath{\mathbf{P}}}
\newcommand{\LL}[1]{{\ensuremath{\mathbf{L}(#1)}}}
\newcommand{\RR}[1]{{\ensuremath{\mathbf{R}(#1)}}}

\newcommand{\PL}{{\ensuremath{\mathbf{P}_{\!\sf L}}}}
\newcommand{\PR}{{\ensuremath{\mathbf{P}_{\!\sf R}}}}

\newcommand{\HH}[1]{\ensuremath{\mathbf{H}_{#1}}}
\newcommand{\HHt}[1]{\ensuremath{\mathbf{H}_{#1}(\theta)}}

\newcommand{\T}{\ensuremath{T}}
\newcommand{\Tt}{\ensuremath{T(\theta)}}

\newcommand{\Tc}{\ensuremath{T_{{\rm opt}}}}
\newcommand{\Tz}{\ensuremath{T(0^{\circ})}}
\newcommand{\Tte}{\ensuremath{T(\theta+\varepsilon)}}

\newcommand{\SP}[1]{\ensuremath{S(#1)}}

\newcommand{\e}{\ensuremath{e}}
\newcommand{\ex}{\ensuremath{\bar{e}}}
\renewcommand{\c}{\ensuremath{c}}

\newcommand{\lx}{\ensuremath{{l}}}
\newcommand{\s}{\ensuremath{s}}

\newcommand{\cf}{\ensuremath{c_{\sf f}}}
\newcommand{\cb}{\ensuremath{c_{\sf b}}}

\newcommand{\g}{\ensuremath{g}}
\newcommand{\gt}{\ensuremath{g}}
\newcommand{\gte}{\ensuremath{g_{\varepsilon}}}
\newcommand{\gpt}{\ensuremath{g'}}
\newcommand{\gppt}{\ensuremath{g''}}

\newcommand{\Ct}{\ensuremath{C(\theta)}}
\newcommand{\Cte}{\ensuremath{C(\theta+\varepsilon)}}

\newcommand{\p}{\ensuremath{p}}
\newcommand{\q}{\ensuremath{q}}

\renewcommand{\v}{\ensuremath{v}}
\newcommand{\vg}{\ensuremath{\v_{\g}}}
\newcommand{\vgt}{\ensuremath{\v_{\gt}}}

\newcommand{\vgpt}{\ensuremath{\v_{\gpt}}}
\newcommand{\vgppt}{\ensuremath{\v_{\gppt}}}

\newcommand{\upt}{\ensuremath{\u'}}

\newcommand{\upte}{\ensuremath{\u'_{\varepsilon}}}

\renewcommand{\u}{\ensuremath{u}}

\newcommand\blfootnote[1]{%
  \begingroup
  \renewcommand\thefootnote{}\footnote{#1}%
  \addtocounter{footnote}{-1}%
  \endgroup}
  
\title{Shortest Watchman Tours in Simple Polygons under Rotated Monotone Visibility}
\date{}

\author[1]{Bengt J.~Nilsson\thanks{Email: bengt.nilsson.TS@mau.se}}
\author[2]{David Orden \thanks{Email: david.orden@uah.es}}
\author[3]{Leonidas Palios\thanks{Email: palios@cs.uoi.gr}}
\author[4]{Carlos Seara\thanks{Email: carlos.seara@upc.edu}}
\author[5]{Pawe\l{} \.Zyli\'nski\thanks{Email: pawel.zylinski@ug.edu.pl}}

\affil[1]{Malm\"o University, Sweden.}
\affil[2]{Universidad de Alcal\'{a}, Spain.}
\affil[3]{University of Ioannina, Greece.}
\affil[4]{Universitat Polit\`{e}cnica de Catalunya, Spain.}
\affil[5]{University of Gda\'nsk, Poland.}

\begin{document}
  \maketitle
 \blfootnote{
 	\begin{minipage}[l]{0.3\textwidth}
 		\includegraphics[trim=10cm 6cm 10cm 5cm,clip,scale=0.12]{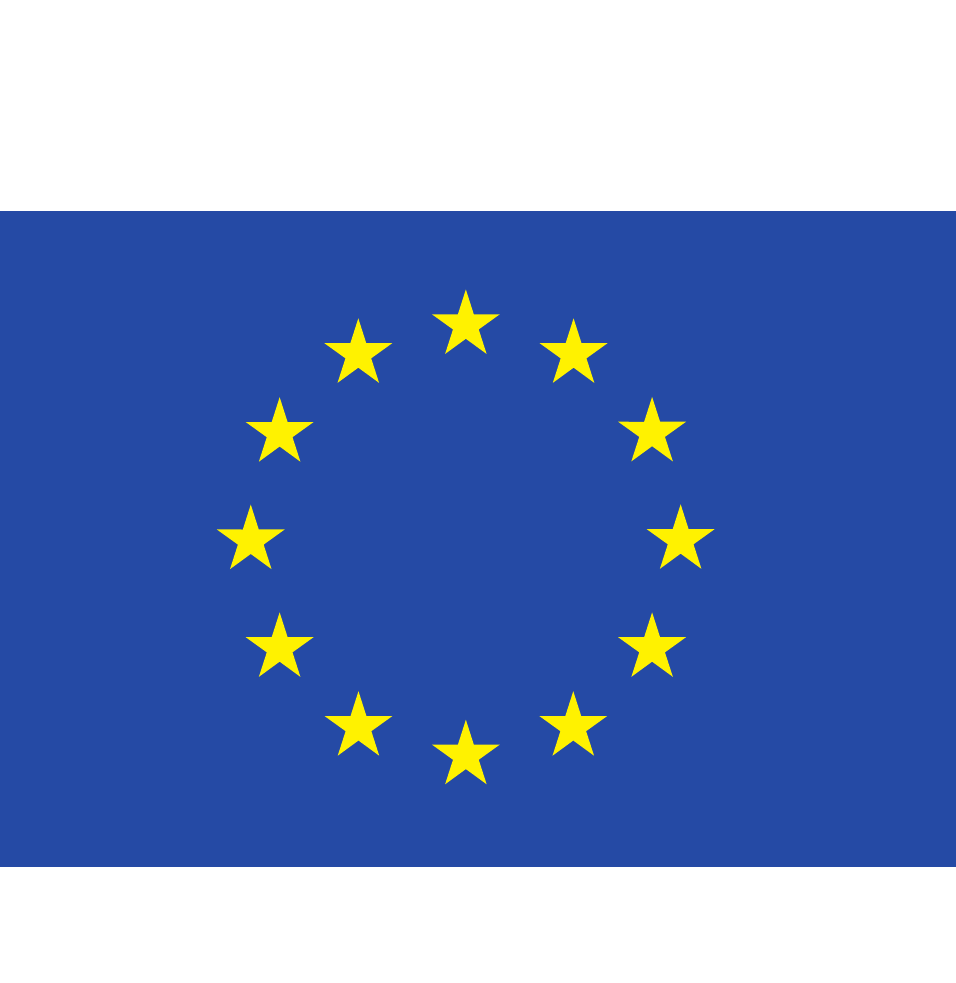}
 	\end{minipage}
 	\hspace{-3cm}
 	\begin{minipage}[l][1cm]{0.75\textwidth}
       	This work has received funding from the European Union's Horizon 2020
       	research and innovation programme under the Marie Sk\l{}odowska-Curie
       	grant agreement No 734922.
      \end{minipage}
 }
 \vspace{-0.9cm}

\vspace{2ex}
\begin{abstract}
\noindent We present an $O(nrG)$ time algorithm for computing and maintaining a minimum length shortest watchman tour that sees a simple polygon under monotone visibility in direction $\theta$, while $\theta$  varies in $[0,180^{\circ})$, obtaining the directions for the tour to be the shortest one over all tours, where $n$ is the number of vertices, $r$ is the number of reflex vertices, and $G\leq r$ is the maximum number of gates of the polygon used at any time in the algorithm.

\end{abstract}

\section{Introduction}\label{sec:intro}

Arguably, problems concerning visibility and motion planning in polygonal environments are among the most well-studied in computational geometry. A problem that encompasses both the visibility and motion planing aspects is that of computing a~{\em shortest watchman tour\/} in an environment, i.e., the shortest closed tour that sees the complete free space of the environment. This problem, under the standard visibility model where two points see each other if the line segment connecting them does not intersect the exterior of the environment free space, has been shown NP-hard~\cite{ChiNta:rectwatchman,DumTot:holewatchman} and even $\Omega(\log n)$-inapproximable~\cite{Mit:approxwatchman} for polygons with holes having a~total of $n$ segments.
With respect to simple polygons, Chin and Ntafos~\cite{ChiNta:rectwatchman} showed a linear time algorithm to compute a shortest watchman tour in a simple rectilinear polygon. Then, after a few false starts~\cite{ChiNta:watchman,HamNil-SWR:conf,TanHir:watchman,TanHirIna:watchman}, Tan~{\em et~al}.~\cite{TanHirIna:corrwatchman} proved an $O(n^4)$ time dynamic programming algorithm for computing a shortest watchman tour through a given boundary point in an arbitrary simple polygon, the so-called {\em fixed\/} watchman tour. This was later improved by Dror~{\em et al}.~\cite{DroEfrLubMit:watchman} to  $O(n^3\log n)$ time. Carlsson~{\em et~al}.~\cite{CarJonNil:journal} showed how to generalize algorithms for a shortest fixed watchman tour to compute a shortest watchman tour in a simple polygon without any pre-specified point to pass through, a {\em floating\/} watchman tour, using quadratic factor overhead. Tan~\cite{Tan:watchman} improved this to a linear factor overhead, thus establishing an $O(n^4\log n)$ time for the floating case, if combined with the algorithm by Dror~{\em et al}.~\cite{DroEfrLubMit:watchman}.

Imagine now that a navigating robot is capable only to see along paths restricted to some (predefined) orientations, with respect to the coordinate system of its free space map. Since changing this coordinate system may affect an optimal trajectory of the robot, a natural variation on the aforementioned watchman tour problem is to allow, as an interrelated optimization process, rotation of the coordinate system (equiv.\ the input polygon), in order to obtain a better solution. Whereas under standard visibility, while rotating the coordinate system, the length of its shortest watchman tour is an invariant, under a non-standard one, it may not be the case,
and moreover, the lengths of such tours may vary arbitrarily. Therefore in this paper, inspired also by some already existing work on the issue of optimizing with rotation~\cite{AleOrPaSeUr:capturingpoints,AleOrSeUr:convexhull,OrPaSeUrZy:widestcorridor}, we introduce the problem of computing the orientation $\theta$ that minimizes the length of a watchmen tour taken over all rotations of the coordinate system, under monotone visibility~\cite{AsaGhoShe:inhandbook}. 
In particular, we present an $O(nrG)$ time algorithm for computing and maintaining a shortest (floating) watchman tour that sees a simple polygon under monotone visibility in direction $\theta$, while $\theta$  varies in $[0,180^{\circ})$, obtaining the directions for the tour to be the shortest one over all tours, where $n$ is the number of vertices, $r$ is the number of reflex vertices, and $G$ is the maximum number of gates of the polygon used at any time in the algorithm; see Section~\ref{sec:prel} for a formal definition. In particular, we have $G\leq r$ in all cases.



Visibility plays a central role in diverse advanced application areas, for example, in surveillance, computer graphics, sensor placement, and motion planning as well as in wireless communication. Therefore, our particular monotone visibility model has practical applications in material
processing and manufacturing. We also note in passing that the problem of computing the shortest watchman tour for a given polygon, under rotated monotone visibility, is related to a variety of problems concerning the concept of ``oriented'' kernels in polygons, which has already attracted attention in the literature. In particular, for a given set ${\cal O}$ of predefined directions, Schuierer et al.~\cite{SRW} provided an~algorithm to compute the $\mathcal{O}$-Kernel of a simple polygon. Next,  Schuierer and Wood~\cite{SW} introduced the concept of  the {\em external} $\mathcal{O}$-Kernel of a polygon, in order to compute the $\mathcal{O}$-Kernel of a simple polygon with holes. In addition, when restricted to ${\cal O}=\{0^\circ,90^\circ\}$, Gewali~\cite{G} described a linear algorithm for orthogonal polygons without holes, and a quadratic one for orthogonal polygons with holes, whereas Palios~\cite{P} gave an output-sensitive algorithm for that problem in orthogonal polygons with  holes. More recently, Orden et al.~\cite{EuroCG18} presented algorithms for computing the orientations $\theta$ in $[-90^\circ,90^\circ)$ such that the $\{\theta\}$-Kernel of a simple (or orthogonal) polygon is not empty, has maximum/minimum area or maximum/minimum perimeter.

Our work is organized as follows. Section~\ref{sec:prel} contains preliminary results and an overview of the linear-time algorithm for computing a  shortest watchman tour in rectilinear polygons, being the basis for our approach. In Section~\ref{sec:algorithm}, we present our algorithm and prove its correctness. Finally, in Section~\ref{sec:analysis}, we analyze the running time of the algorithm and conclude the presentation in Section~\ref{sec:conc}.
\section{Preliminaries}\label{sec:prel}
Let $\theta$ be a direction, specified using its angle to the $x$-axis, and let \P\ be a simple polygon having $n$ edges. We define a path $\Pi$ inside \P\ to be {\em monotone w.r.t.\@ direction $\theta$\/} or {\em $\theta$-monotone\/} if and only if the intersection between any line parallel to direction $\theta$ and $\Pi$ is a~connected set. In standard visibility, two points \p\ and \q\ in \P\ see each other if and only if the line segment between \p\ and \q\ does not intersect the exterior of~\P. In $\theta$-monotone visibility, two points \p\ and \q\; {\em $\theta$-see} each other if and only if there is a $\theta$-monotone path between \p\ and \q\ not intersecting the exterior of~\P.

Let \v\ be a reflex vertex of \P\ incident to the boundary edge \e. If we extend~\e\ maximally inside \P, we obtain a line segment \ex\ collinear with \e, and we can associate the direction to \ex\ being the same as that of \e\ as we traverse the boundary of \P\ in counterclockwise order; see Figure~\ref{fig:Tmustcross}(a). Therefore, any reflex vertex is adjacent to two extensions in~\P. Any directed segment in \P\ intersecting the interior of \P\ and connecting two boundary points of \P\ is called a {\em cut\/} of \P. Thus, an extension \ex\ is a cut in \P. A cut \c\ partitions \P\ into two components, \LL{\c}, the component locally to the left of \c\ according to its direction, and \RR{\c}, the component locally to the right of \c\ according to its direction. Consider now a closed tour \T\ inside \P. Chin and Ntafos~\cite{ChiNta:rectwatchman,ChiNta:watchman} argue for standard visibility that in order for \T\ to see the whole polygon \P, it is sufficient for \T\ to see the vertices of \P, and therefore, it is sufficient for \T\ to have a point in~\LL{\ex} for every extension \ex\ in~\P, i.e., to intersect the left of any extension in~\P; again see Figure~\ref{fig:Tmustcross}(a).

\begin{figure}
\begin{center}
\includegraphics[width=5in]{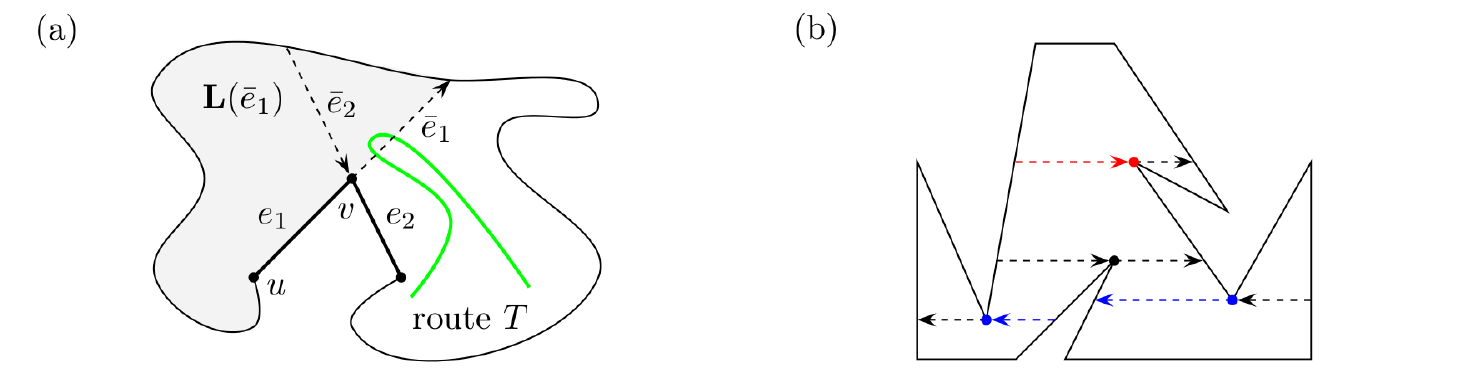}
\caption{(a) $T$ must have a point in $\protect\LL{\bar{e}_1}$ (marked gray) to see $u$. (b) Dominant extensions (gates) are marked with blue and red.}\label{fig:Tmustcross}
\end{center}
\end{figure}

Next, for monotone visibility in direction $\theta$, we introduce the following definitions. Consider a reflex vertex \v\ of \P\ and let \lx\ be the directed line with angle $\theta$ to the $x$-axis passing through \v. If the two boundary edges incident to \v\ lie on the same side of \lx, let \s\ be the maximal segment collinear to \lx\ inside \P\ that passes through \v. The segment \s\ partitions~\P\ into three components incident to \v. Two components \PL\ and~\PR\ have \v\ as a convex vertex and the third component has \s\ as a boundary edge; \PR\ is the first subpolygon traversed by the counterclockwise traversal of~\P\ starting at \v, whereas \PL\ is the other component having~$v$ as convex vertex. The problem thus reduces to obtaining the shortest tour that intersect the left of a given set of cuts.

Now, we can argue similarly for monotone visibility in direction $\theta$ as in the standard visibility case. Assume the tour \T\ has points in \PR. Unless~\T\ intersects~\s, \T~cannot see any points in \PL. Similarly, if \T\ has points in \PL, then \T\ cannot see any points in \PR, unless it intersects \s. To mimic the standard visibility situation, we introduce two cuts \cf\ and \cb\ incident to \v, where \cf\ is the portion of \s\ bounding \PR, and \cb\ is the portion of \s\ bounding \PL. Specifically, the cut~\cf\ is directed away from \v\ and we call it a {\em forward $\theta$-cut\/},  and symmetrically, the cut \cb\ is directed towards~\v\ and we call it a {\em backward $\theta$-cut}; see Figure~\ref{fig:Tmustcross}(b).

We also color the $\theta$-cuts and their associated reflex vertices. A $\theta$-cut is {\em red}, if the boundary edges incident to the associated vertex \v\ both lie locally to the right of the directed line \lx\ defined above. The vertex \v\ is thus called a~{\em red\/} vertex. Analogously, a $\theta$-cut is {\em blue}, if the boundary edges incident to the associated vertex \v\ both lie locally to the left of the directed line \lx, and \v\ is called a {\em blue\/} vertex. Other reflex vertices are not colored as they do not break monotonicity w.r.t.\@ direction $\theta$ and are therefore not used; see Figure~\ref{fig:Tmustcross}(b).

Similarly to the standard visibility case, we define the region \LL{\c} for a $\theta$-cut \c\ to be the part of the polygon \P\ locally to the left of \c\ according to its direction, and \RR{\c} to be the part of the polygon \P\ locally to the right of \c. We claim the following lemma that corresponds to the standard visibility case~\cite{ChiNta:watchman}; its correctness follows from the definition of the $\theta$-cut.
\begin{lemma}
A tour \T\ in \P\ is a watchman tour under $\theta$-monotone visibility if and only if it intersects the region \LL{\c}, for every $\theta$-cut~\c\ in~\P.
\end{lemma}
This lemma allows us to use the algorithm of Chin and Ntafos~\cite{ChiNta:rectwatchman} for computing the shortest watchman tour under $\theta$-monotone visibility, since it computes the shortest tour that intersect the left of a set of cuts. The algorithm works roughly as follows. First, it identifies the proper set of cuts inside the polygon. Second, it reduces the shortest tour problem to a shortest path problem in a triangulated two-manifold, computes the shortest path, and transforms the path to a tour in the original polygon. For the classical watchman tour in a simple polygon, the first step can be done with a ray shooting data structure~\cite{GuiHerLevShaTar:visandshortpath,HerSur:rayshoot} in total $O(n\log n)$ time. In our case, since all $\theta$-cuts are parallel line segments, we can do the first step in $O(n)$ time using the algorithm by Chazelle~\cite{Cha:triangulation} that partitions a simple polygon into $O(n)$ visibility trapezoids by introducing parallel line segments at reflex vertices in the polygon.

Similarly as for standard visibility, we define a {\em dominance relation\/} between cuts. Given two $\theta$-cuts \c\ and $\c'$, we say that \c\ {\em dominates\/} $\c'$ if $\LL{\c}\subset\LL{\c'}$. We call \c\ a~{\em dominating\/} $\theta$-cut or {\em gate}, if \c\ is not dominated by any other $\theta$-cut in \P\!. We refer to the issuing reflex vertex of a gate as the {\em gate vertex\/}, whereas the edge touched by the other endpoint of the gate is called the {\em gate edge}. Carlsson~{\em et al.}~\cite{CarJonNil:journal} show how to compute the dominating cuts in \P\ with the standard visibility in $O(n)$ time, given the complete set of cuts ordered along the boundary. Their method transfers directly to the case of $\theta$-monotone visibility, since given the trapezoidation of the polygon, the ordering can be obtained by a~traversal of the boundary, ordering the forward $\theta$-cuts and the backward $\theta$-cuts separately, and then merging these two sets of $\theta$-cuts. The process thus takes $O(n)$ time in total.

In the algorithm, we establish {\em gates\/} as explained above and remove the portions of the polygon that lie locally to the left of them, resulting in the polygon~$\P'(\theta)$; see Figure~\ref{fig:unrollingexample}(a). The optimal tour will only reflect on the gates in~\P\ to see everything on the other side of them, so it is completely contained in~$\P'(\theta)$. We then triangulate $\P'(\theta)$ and establish a constant-size subset~\VV\ of vertices such that the optimal tour must pass through at least one of them; see Section~\ref{sec:analysis}. In the next step, we compute for each vertex \v\ in \VV, a triangulated two-manifold~$\HH{\v}(\theta)$ (see Section~\ref{subsec:2man}) such that the shortest path from \v\ to its image $\v'$ in~$\HH{\v}(\theta)$ corresponds to the shortest watchman tour in \P\ that passes through \v. We then establish the shortest path \SP{\v,\v'}\ in~$\HH{\v}(\theta)$, for each \v\ in \VV, pick the shortest of these paths, and finally transform it back to the polygon $\P'(\theta)$. The whole computation can be done in $O(n)$ time
~\cite{ChiNta:rectwatchman}.
\subsection{The Two-Manifold}\label{subsec:2man}
In our approach, we extensively exploit the concept of the two-manifold used in the algorithm by Chin and Ntafos~\cite{ChiNta:rectwatchman}. We therefore provide some more details.

Without loss of generality assume that $\theta=0^\circ$. Consider the dominant extensions (or gates in our case) in \P. Chin and Ntafos~\cite{ChiNta:rectwatchman,ChiNta:watchman} prove that an optimal watchman tour will never intersect the interior of any region \LL{\c} for any dominant extension~\c. Thus, we define $\P'={ \P'(0^\circ)}\defeq\big(\P\setminus\bigcup_{\c\in\GG}\LL{\c}\big)^*$, where \GG\ is the set of dominant extensions and $\SS^*$ denotes the closure of a set \SS. (We take the closure to include the boundary points of $\P'$\!.\,) The dominant extensions (gates) of \P\ are now part of the boundary of $\P'$\!, and so we refer to them as the {\em essential edges} of $\P'$\!. The polygon $\P'$ is triangulated and, given a vertex~\v, $\P'$ is then {\em unrolled\/} from \v\ to $\v'$ using the essential edges as mirrors giving \HH{\v} \big($=\HH{\v}(0^\circ)$\big), where $\v'$ is the image of~\v, as follows. The counterclockwise traversal of the boundary of $\P'$ starting at \v\ encounters the incident triangles from the triangulation in order along the edges and vertices of the traversal. When reaching an essential edge, we reflect all subsequent triangles using the essential edge as mirror. As the traversal continues, we repeat this step until the traversal reaches the vertex \v\ again. In the two-manifold, the second instance of the vertex \v\ is called the {\em image\/} of \v\ and denoted $\v'$.  Between the vertex \v\ and the first essential edge, subsequent consecutive pairs of essential edges, and between the final essential edge and $\v'$, we perform a standard breadth-first-search to keep those triangles of the triangulation that form a path between the mirroring segments in \HH{\v} (see Figure~\ref{fig:unrollingexample}), as this will aid the shortest path finding algorithm in later steps~\cite{GuiHerLevShaTar:visandshortpath,LeePre:shortpath}. The size of \HH{\v}\ is linear, because each triangle from the triangulation of $\P'$ is used at most six times in the construction of \HH{\v}, once for each side of the triangle and once for each vertex of the triangle, since each boundary edge and vertex is passed only once as we perform the traversal of~$\P'$\!.

From Heron~\cite{Her:catoptrica}, we know that the shortest path between two points that also touches a line (with both points on the same side of the line) makes perfect reflection on the line; see Figure~\ref{fig:unrollingexample}(b,\,c). Thus, the shortest path in \HH{\v}\ between \v\ and its image $\v'$ corresponds exactly to the shortest watchman tour in \P\ that passes through \v\ as the path is folded back along the gates in \P. Computing the shortest path in a triangulated simple polygon can be done in linear time~\cite{GuiHerLevShaTar:visandshortpath,LeePre:shortpath}, and since \HH{\v}\ consists of a linear number of connected triangles, computing the path takes $O(n)$ time. Folding back the path to obtain the tour also takes $O(n)$ time by traversing the path and computing the points of intersection between the path and the triangle sides corresponding to gates in~\P.
\begin{figure}
\begin{center}
\includegraphics[width=5.75in]{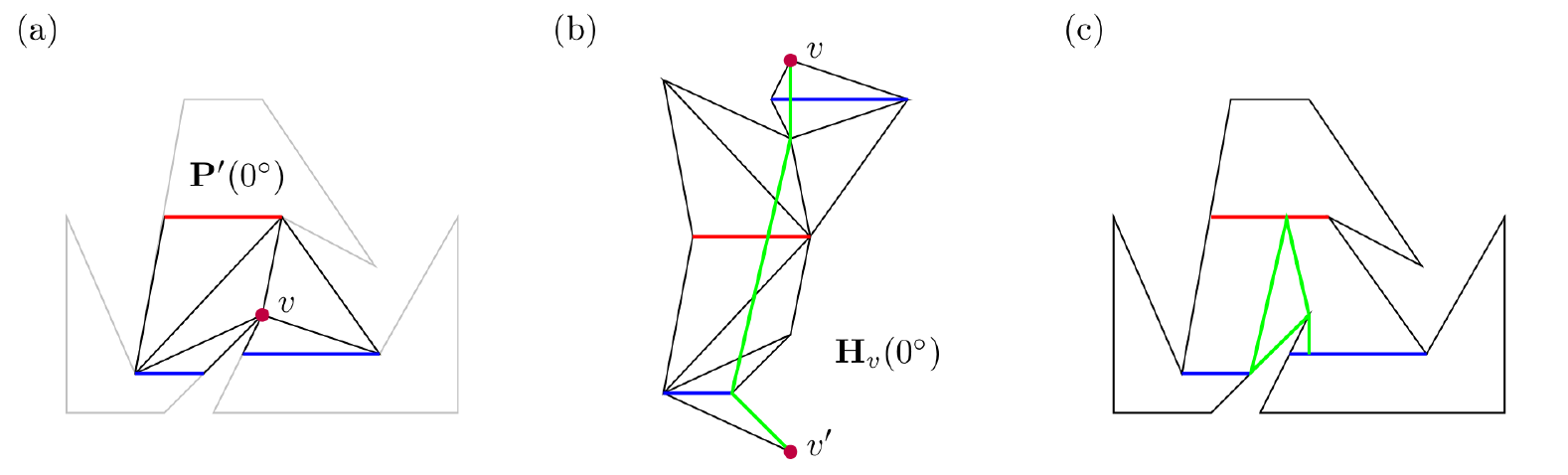}
\caption{(a,\,b) The polygon $\P'(0^\circ)$ with the distinguished vertex $v$ that acts as starting point for the unrolling process resulting in the two-manifold $\HH{\v}(0^\circ)$.
(c)~The resulting shortest watchman tour in~$\protect\P$.}\label{fig:unrollingexample}
\end{center}
\end{figure}
\section{The Algorithm}\label{sec:algorithm}
Given that we can compute the shortest watchman tour under monotone visibility in a specific direction $\theta$ in linear time, our objective is to find the direction~$\theta$ for which the length of a shortest watchman tour under monotone visibility in this direction is minimal. Let \Tt\ denote a shortest watchman tour under monotone visibility in direction $\theta$. The idea of the algorithm is to compute the tour \Tz, and then rotate the direction $\theta$ from $0^{\circ}$ to $180^{\circ}$, updating \Tt\ as the rotation proceeds.

Consider the tour \Tt\ for a fixed direction $\theta$ and let \GGt\ be the set of gates visited by~\Tt. Treat \Tt\ as a (weakly simple) polygon and divide the vertex set of \Tt\ into two types. The {\em stable vertices\/} of \Tt\ coincide with reflex vertices of~\P, even if they sometimes correspond to convex vertices in \Tt; see Figure~\ref{fig:unrollingexample}(c). The {\em moving vertices\/} of \Tt\ are the reflections on the gates in \GGt. We partition \Tt\ into subpaths going from one stable vertex to the next one in counterclockwise order along \Tt. Any such path is either a line segment between stable vertices of \Tt\ or it is a path that starts at a stable vertex, passes a consecutive sequence of moving vertices, and finishes then again at a~stable vertex. We call such a subpath a {\em maximal moving subpath of \Tt}. Any maximal moving subpath of \Tt\ has the following property.
\begin{lemma}\label{lem:onlythreegates}
A maximal moving subpath \Ct\ of \Tt\ has at most three moving vertices and they touch gates in order having alternating colors.
\end{lemma}
\begin{proof}
Let \v\ be the stable vertex at the first endpoint of \Ct. Let \SP{\v,\v'}\ be the shortest path between \v\ and its image $\v'$ in the two-manifold \HHt{\v}. The moving vertices of \Ct\ correspond to consecutive crossings of gates by \SP{\v,\v'}\ in \HHt{\v}\ without touching a stable vertex. Since all gates are parallel, when following \SP{\v,\v'}, no two consecutive gates in \HHt{\v}\ can have the same color without \SP{\v,\v'}\ (and thus \Tt) touching a stable vertex, otherwise one of them dominates the other, contradicting that they are both gates. Thus, the sequence of consecutive gates in \HHt{\v}\ is color alternating. Next, it is clear that the sequence of gates cannot consist of more than three gates, since four or more would mean that \Tt\ is self intersecting and thus could be shortened~\cite{ChiNta:rectwatchman,ChiNta:watchman}. 
\end{proof}

Given \Tt\ and the set of gates \GGt, assume we increase the rotation to $\theta+\varepsilon$ to obtain the tour \Tte\ and the set of gates \GGte; we refer to such a rotation as an {\em $\varepsilon$-rotation}. We say that \Tt\ and \Tte\ are {\em close} if each of the following properties hold:
\begin{enumerate}
    \item The stable vertices of \Tt\ and \Tte\ are the same.
    \item The gate vertices for the gates in \GGt\ and \GGte\ in \P\ are the same.
    \item For any pair of gates $\gt\in\GGt$ and $\gte\in\GGte$ with the same gate vertex, they also have the same gate edge.
    \item For any pair of gates $\gt\in\GGt$ and $\gte\in\GGte$ with the same gate vertex~\vgt, if \Tt\ touches \vgt, then \Tte\ also touches \vgt, and if \Tt\ touches the other endpoint of \gt, then \Tte\ also touches the other end point of~\gte.
\end{enumerate}
We claim the following lemma.

\begin{lemma}\label{lem:smoothrotation}
\!\!If \Tt\ and \Tte\ are close, then $\len{\Tte}\!\!=\!\!\len{\Tt} + \sum_{k=1}^{|\GGt|} f_k(\varepsilon)$, where

\begin{small}
\begin{align*}
f_k(\varepsilon)
    &=
    \sqrt{\frac{a^2_{k,0} + a_{k,1}\tan\varepsilon+a_{k,2}\tan^2\varepsilon}{1 + a_{k,3}\tan\varepsilon+a_{k,4}\tan^2\varepsilon}}\! -\! a_{k,0}
    +
    \sqrt{\frac{a^2_{k,5} + a_{k,6}\tan\varepsilon+a_{k,7}\tan^2\varepsilon}{1 + a_{k,8}\tan\varepsilon+a_{k,9}\tan^2\varepsilon}}\! -\! a_{k,5}
    \nonumber\\
    & \qquad\quad +
    \sqrt{\frac{a^2_{k,10} + a_{k,11}\tan\varepsilon + a_{k,12}\tan^2\varepsilon + a_{k,13}\tan^3\varepsilon + a_{k,14}\tan^4\varepsilon}{1 + a_{k,15}\tan\varepsilon + a_{k,16}\tan^2\varepsilon + a_{k,17}\tan^3\varepsilon + a_{k,18}\tan^4\varepsilon}} - a_{k,10}
    \nonumber\\
    & \qquad\quad +
    \sqrt{\frac{a^2_{k,19} + \sum_{i=1}^{14} a_{k,19+i}\tan^i\varepsilon}{1 + \sum_{i=1}^{14} a_{k,33+i}\tan^i\varepsilon}} - a_{k,19},
\end{align*}
\end{small}for some constants $a_{k,0},\ldots,a_{k,47}$, $1\leq k\leq|\GGt|$, only depending on the stable vertices, the gate vertices, the gate edges, and the angle~$\theta$.
\end{lemma}

\begin{proof}
Let \Ct\ be a maximal moving subpath of \Tt. We consider three cases, depending on the number (at most three by Lemma~\ref{lem:onlythreegates}) of gates touched by \Ct.

\medskip
\noindent Case 1: {\em \Ct\ touches only one gate \gt}. Assume without loss of generality that \gt\ is red and has gate vertex \vgt, the other case is completely symmetric; see Figure~\ref{fig:smoothrotation}(a). Let \u\ and \v\ be the endpoints of \Ct, where \v\ is reached before \u\ along a counterclockwise traversal of \Tt\ starting at a point not on~\Ct.

\begin{figure}[t]
\begin{center}
    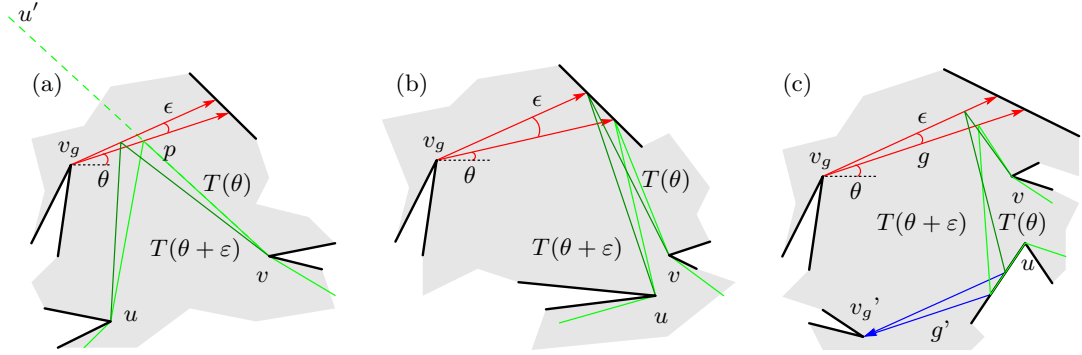
    \caption{Illustrating the proof of Lemma~\protect\ref{lem:smoothrotation}.}\label{fig:smoothrotation}
\end{center}
\end{figure}

If \Ct\ does not touch any endpoint of $\gt=\g(\theta)$, then by reflecting \u\ along~\gt, we obtain the point \upt, where the coordinates only depend on \u\ and \gt. Thus, $\len{\Ct}=\len{\v,\upt}=\len{\v,p}+\len{p,\u}$, where $p$ is the intersection between $[\v,\upt]$ and~\gt. Similarly, $\len{\Cte}=\len{\v,\upte}$, where \upte\ is the reflection of \u\ along \gte; see Figure~\ref{fig:smoothrotation}(a). By routine trigonometry, we have

\begin{align}\label{eqn:oneperfect}
    f(\varepsilon)
    &=
    \len{\Cte}-\len{\Ct} = \sqrt{\frac{a^2_0 + a_1\tan\varepsilon+a_2\tan^2\varepsilon}{1 + a_3\tan\varepsilon+a_4\tan^2\varepsilon}} - a_0,
\end{align}
for some constants $a_0,\ldots,a_4$ only depending on the points \u, \v, \vgt, and the angle~$\theta$.

Assume next that \Ct\ touches an endpoint of \gt. If that endpoint is \vgt, then $\len{\Ct}=\len{\Cte}$, and so $f(\varepsilon)=0$. If \Ct\ touches the other endpoint of \gt, then

\begin{align}\label{eqn:onetouch}
    f(\varepsilon)
    &=
    \sqrt{\frac{a^2_0 + a_1\tan\varepsilon+a_2\tan^2\varepsilon}{1 + a_3\tan\varepsilon+a_4\tan^2\varepsilon}} - a_0
    +
    \sqrt{\frac{a^2_5 + a_6\tan\varepsilon+a_7\tan^2\varepsilon}{1 + a_8\tan\varepsilon+a_9\tan^2\varepsilon}} - a_5,
\end{align}
for some constants $a_0,\ldots,a_9$ only depending on the points \u, \v, \vgt, the gate edge, and the angle~$\theta$; see Figure~\ref{fig:smoothrotation}(b).

\medskip
\noindent Case 2: {\em \Ct\ touches two gates \gt\ and \gpt}. First, it follows from the proof of Lemma~\ref{lem:onlythreegates} that \gt\ and \gpt\ cannot have the same color. Next, consider the segment~\s\ of \Ct\ connecting \gt\ and~\gpt. \Ct\ cannot have its endpoints on either side of~\s\ \big(in $\P'(\theta)$\big), as otherwise \Tt\ must self intersect, a contradiction~\cite{ChiNta:watchman}. Thus, both \u\ and \v\ lie on the same side of \s, whereby \s\ must touch the boundary of \P, as otherwise \Tt\ could be made shorter by shifting \s\ along the gates \gt\ and \gpt\ towards \u\ and \v, another contradiction. Since our assumption is that \Ct\ does not touch a stable vertex in its interior (a reflex vertex in \P), \s\ must touch at least one endpoint of one of the gates \gt\ or \gpt. If this endpoint is a gate vertex, we can view this gate vertex as a~fixed point of \Tt\ (or stable vertex) and either $f(\varepsilon)=0$, if also the other endpoint of \s\ touches a gate vertex, or one of Equalities~(\ref{eqn:oneperfect}) and~(\ref{eqn:onetouch}) hold by the previous argument.

If \s\ touches the boundary of \P\ at a gate edge, we have two cases. Either, the other endpoint of \s\ touches the interior of the other gate or it also touches a gate edge. In the first case, we have by routine trigonometry that

\begin{align}\label{eqn:twoperfect}
    f(\varepsilon)
    &=
    \sqrt{\frac{a^2_0 + a_1\tan\varepsilon+a_2\tan^2\varepsilon}{1 + a_3\tan\varepsilon+a_4\tan^2\varepsilon}} - a_0
    +
    \sqrt{\frac{a^2_5 + \sum_{i=1}^{12} a_{5+i}\tan^i\varepsilon}{1 + \sum_{i=1}^{12} a_{17+i}\tan^i\varepsilon}} - a_5,
\end{align}
for some constants $a_0,\ldots,a_{29}$ only depending on the points \u, \v, \vgt, \vgpt, the touched gate edge, and the angle~$\theta$; see Figure~\ref{fig:smoothrotation}(c). In the second case, we have two endpoints of three segments that both touch gate edges, and we obtain

\begin{align}\label{eqn:twotouch}
    f(\varepsilon)
    &=
    \sqrt{\frac{a^2_0 + a_1\tan\varepsilon+a_2\tan^2\varepsilon}{1 + a_3\tan\varepsilon+a_4\tan^2\varepsilon}} - a_0
    +
    \sqrt{\frac{a^2_5 + a_6\tan\varepsilon+a_7\tan^2\varepsilon}{1 + a_8\tan\varepsilon+a_9\tan^2\varepsilon}} - a_5
    \nonumber\\
    &+
    \sqrt{\frac{a^2_{10} + a_{11}\tan\varepsilon + a_{12}\tan^2\varepsilon + a_{13}\tan^3\varepsilon + a^2_{14}\tan^4\varepsilon}{1 + a_{15}\tan\varepsilon + a_{16}\tan^2\varepsilon + a_{17}\tan^3\varepsilon + a^2_{18}\tan^4\varepsilon}} - a_{10},
\end{align}
for some constants $a_0,\ldots,a_{18}$ only depending on the points \u, \v, \vgt, \vgpt, the two touched gate edges, and the angle~$\theta$.

\medskip
\noindent Case 3: {\em \Ct\ touches three gates \gt, \gpt, and \gppt}. Again, observe that they must have alternating color (see the proof of Lemma~\ref{lem:onlythreegates}). Next, by the same argument as in the previous case, at least two segments of \Ct\ connecting the three gates must touch the boundary of \P. If such a touching point is a gate vertex, this reduces to one of Equalities~(\ref{eqn:oneperfect})--(\ref{eqn:twotouch}) since we can view this vertex of \Tt\ as stable. If \Ct\ touches the middle gate in the interior, then again by routine trigonometry we obtain

\begin{align}\label{eqn:threeperfect}
    f(\varepsilon)
    &=
    \sqrt{\frac{a^2_0 + a_1\tan\varepsilon+a_2\tan^2\varepsilon}{1 + a_3\tan\varepsilon+a_4\tan^2\varepsilon}} - a_0
    +
    \sqrt{\frac{a^2_5 + a_6\tan\varepsilon+a_7\tan^2\varepsilon}{1 + a_8\tan\varepsilon+a_9\tan^2\varepsilon}} - a_5,
    \nonumber\\
    & \qquad\quad +
    \sqrt{\frac{a^2_{10} + \sum_{i=1}^{14} a_{10+i}\tan^i\varepsilon}{1 + \sum_{i=1}^{14} a_{24+i}\tan^i\varepsilon}} - a_{10},
\end{align}
for some constants $a_0,\ldots,a_{38}$ only depending on the points \u, \v, \vgt, \vgpt, \vgppt, the two touched gate edges, and the angle~$\theta$. Finally, if \Ct\ touches the middle gate at its gate edge, then we obtain

\begin{align}\label{eqn:threetouch}
    f(\varepsilon)
    &=
    \sqrt{\frac{a^2_0 + a_1\tan\varepsilon+a_2\tan^2\varepsilon}{1 + a_3\tan\varepsilon+a_4\tan^2\varepsilon}} - a_0
    +
    \sqrt{\frac{a^2_5 + a_6\tan\varepsilon+a_7\tan^2\varepsilon}{1 + a_8\tan\varepsilon+a_9\tan^2\varepsilon}} - a_5
    \nonumber\\
    &+
    \sqrt{\frac{a^2_{10} + a_{11}\tan\varepsilon + a_{12}\tan^2\varepsilon + a_{13}\tan^3\varepsilon + a^2_{14}\tan^4\varepsilon}{1 + a_{15}\tan\varepsilon + a_{16}\tan^2\varepsilon + a_{17}\tan^3\varepsilon + a^2_{18}\tan^4\varepsilon}} - a_{10}
    \nonumber\\
    &+
    \sqrt{\frac{a^2_{19} + a_{20}\tan\varepsilon + a_{21}\tan^2\varepsilon + a_{22}\tan^3\varepsilon + a^2_{23}\tan^4\varepsilon}{1 + a_{24}\tan\varepsilon + a_{25}\tan^2\varepsilon + a_{26}\tan^3\varepsilon + a^2_{27}\tan^4\varepsilon}} - a_{19},
\end{align}
for some constants $a_0,\ldots,a_{27}$ only depending on the points \u, \v, \vgt, \vgpt, \vgppt, the three touched gate edges, and the angle~$\theta$. Clearly, combining Equalities~(\ref{eqn:oneperfect})--(\ref{eqn:threetouch}) gives us the lemma.
\end{proof}
By Lemma~\ref{lem:smoothrotation}, as long as \Tt\ maintains the closeness properties in a small neighborhood, $\theta+\varepsilon$ of $\theta$, with $\varepsilon>0$, the length function $\len{\Tt}$ is smooth (continuous and differentiable), and we can obtain the angles of minima for $\len{\Tt}$ using standard analytic methods. Since
the function consists of $O(|\GGt|)$ terms, requiring us to test $O(|\GGt|)$ potential solutions, we can do this in $O(|\GGt|)$ time.

However, when the closeness properties do not hold, at least one of the following changes occur:
the current set of stable vertices of \Tt\ changes, the current set of gate vertices changes, some gate in the current set \GGt\ changes its gate edge, or the tour \Tt\ reaches or leaves an~endpoint of a gate. We call the angles where such changes occur {\em events\/} and present them in further detail next. 

\subsection{Events}\label{sec:events}
In general, we have two types of events: those 
defined by the vertices of the polygon (or pairs of them), and those defined by stable and moving vertices of the current tour. We further subdivide them into the following six types.
%
\begin{description}
    \item[Validity event:] a new gate arises or an old gate disappears. This happens when the gate becomes collinear to a polygon edge adjacent to the gate vertex. 

    \item[Domination event:] a gate ``changes gate vertex'', i.e., a cut \c\ issued from a~vertex \v, previously dominated by a gate \g\ with gate vertex $\v'$, becomes collinear to~\g, \v~and $\v'$ have the same color, and as the rotation proceeds, \v~becomes the new gate vertex.

    \item[Jumping event:] the endpoint of a gate \g\ on the gate edge, reaches a reflex vertex of \P\ issuing a cut of different color to that of~\g.

    \item[Passing event:] the endpoint of a gate on the gate edge, reaches an uncolored reflex vertex or a convex vertex of~\P.

    \item[Bending event:] a maximal moving subpath of \Tt\ reaches or leaves a reflex vertex of~$\P'(\theta)$.

    \item[Cuddle event:] a moving vertex of the tour reaches or leaves a gate endpoint.
\end{description}

\noindent
We have the following lemma.
\begin{lemma}\label{lem:complete}
The set of events is complete.
\end{lemma}
\begin{proof}
Consider the four properties necessary for two tours \Tt\ and \Tte\ to be close. We take the contrapositive for each property and show that the only cases when these can occur is if one of the listed events occurs for~\Tt.
\begin{description}
\item{\em The current set of stable vertices of \Tt\ changes.}
Since the only part of \Tt\ that changes under $\varepsilon$-rotation are the maximal moving subpaths, a stable vertex can never be directly exchanged for another vertex. Therefore, the only other possibilities are that a stable vertex is either added to or removed from \Tt, but these are exactly the {\em bending events}.

\item{\em The current set of gate vertices changes.}
There are three possibilities that this can happen. First, a gate vertex is exchanged for another gate vertex. In order for this to happen, there must be some angle when a gate is collinear to a $\theta$-cut of some other reflex vertex with the same color. These are exactly the {\em domination events}. The two other possibilities are that a gate is either added to or removed from \GGt, but these are exactly the {\em validity events}.

\item{\em Some gate in the current set \GGt\ changes gate edge.}
For this to happen, the endpoint of the gate opposite the gate vertex must lie at a vertex of the polygon. If this vertex is reflex and has the same color as the current gate, we have a~{\em domination event}. If the vertex is reflex but has the opposite color of the gate, by definition, we have a {\em jumping event}. If the vertex is either reflex but uncolored or convex, then we have exactly a {\em passing event}. Since vertices can be of no other types, these three event types cover this case.

\item{\em The tour \Tt\ reaches or leaves an endpoint of a gate.}
This is exactly the definition of the {\em cuddle events}.
\end{description}
%
Thus, the six event types completely cover all the cases.
\end{proof}
The reason for defining six types of events is that our algorithm will handle each of them slightly differently, as is explained in the next subsection.

\subsection{Handling events}\label{sec:eventhandling}
The algorithm maintains, for a given angle $\theta$, the following information: \Tt, \GGt, \HHt{\v} for a reflex vertex \v\ coinciding with a stable vertex of \Tt, the change function $\len{\Tte} =\len{\Tt}+\sum_{1\leq k\leq|\GGt|}f_k(\varepsilon)$
(a function of $\varepsilon$), for each gate in \GGt, the visibility polygon of the gate vertex (for standard straight line visibility, not monotone visibility), and a priority queue \QQ\ maintaining $O(|\GGt|)$ angles of future potential events, a constant number for each gate. We next present how each event is handled during the running of the algorithm as the direction $\theta$ rotates from~$0^{\circ}$ to~$180^{\circ}$\!.

\begin{description}
\item[ALGORITHM]~
\\
Let $\Tc\leftarrow\Tz$, $\theta\leftarrow0^{\circ}$\!, compute the event angles (c.f.~Step~\ref{enum:getangles} in the validity event routine below), insert them in \QQ, and repeat the following steps while $\theta<180^{\circ}$\!.
\par\noindent
Step 1.~~Get the next event angle $\theta$ from the priority queue~\QQ.
\par\noindent
Step 2.~~Depending on the event type, perform one of the following routines as described below.

\smallskip
\item[Validity event routine.] For each such event:
\begin{enumerate}
    \item Compute the gates \GGt\ and the optimal tour \Tt\ for direction $\theta$ in $O(n)$ time as explained in Section~\ref{subsec:2man}.

\item Empty the priority queue~\QQ.

\item For each segment \s\ of the shortest path in \HHt{\v}\ that crosses a gate, we establish the shortest path from the endpoints of the gate to the endpoints of \s, and associate the reflex vertices on those paths that are closest to the endpoints of \s\ (at most four). We can quickly test (in $O(|\GGt|)$ total time) whether any such vertex crosses $s$ during the subsequent rotation and establish at what rotation angle this happens, i.e., the potential next bending and cuddle events.
Insert each of the two events per gate in~\QQ.
\label{enum:nextevent}

\item For each of the gates in \GGt, compute the visibility polygon (for standard straight line visibility) for the gate vertex and obtain the next passing event, jumping event, domination event, and validity event (if they exist) by traversing the boundary of the visibility polygon starting from the angle $\theta$ from the gate vertex. This takes $O(n)$ time per gate.
Insert each of the four events per gate in~\QQ. This takes a total of $O(n|\GGt|)$ time for all the gates.
\label{enum:getangles}




\item Look at the next potential event angle $\theta'$ in \QQ. Compute the change function \len{\Tte}\ and the best local angle $\theta+\varepsilon$, for $\varepsilon>0$, such that $\theta<\theta+\varepsilon\leq\theta'$. If $\len{\Tte}<\Tc$, then update \Tc. This takes $O(|\GGt|)$ time, since the change function has $O(|\GGt|)$ terms, each being the square root of rational polynomials of degree at most~$14$; see Lemma~\ref{lem:smoothrotation}.
\label{enum:update}

\end{enumerate}

\smallskip\noindent
The time complexity is $O(n|\GGt|)$ for this case.

\medskip
\item[Domination event routine.] 
For each such event:
\begin{enumerate}
    \item Update the set of gates in \GGt\ by exchanging one gate for a collinear gate with different gate vertex. This takes constant time.

    \item Remove the events in \QQ\ associated to the old gate. This takes $O(|\GGt|)$ time by a traversal of~\QQ.

    \item For the new gate vertex, compute the visibility polygon around the gate vertex, and obtain the next passing event, jumping event domination event, and validity event (if they exist) by traversing the boundary of the visibility polygon starting from the angle $\theta$ from the gate vertex.
    Insert each of the four events in~\QQ. This takes $O(n)$ time.

    \smallskip
    \item Perform Step~\ref{enum:update} as for validity events.
\end{enumerate}

\smallskip\noindent The time complexity is $O(n)$ for this case.

\medskip
\item[Jumping event routine.] 
For each such event:
\begin{enumerate}
    \item Since the set of gate vertices does not change, only update the gate edge and then recompute the tour obtaining the next jump event by continuing the traversal of the boundary of the visibility polygon starting from the angle $\theta$ from the gate vertex. Add it to \QQ. This takes $O(n)$ time.

    \item Perform Step~\ref{enum:update} as for validity events.
\end{enumerate}

\medskip
\item[Passing event routine.] 
For each such event:
\begin{enumerate}
    \item 
    Update the change function \len{\Tte}\ with the appropriate term consisting of the square root of rational polynomials of degree at most 14. This takes $O(|\GGt|)$ time.

    \item Perform Step~\ref{enum:update} as for validity events.
\end{enumerate}

\medskip
\item[Bending event routine.] Here, neither the set of gate vertices nor the set of gate edges change, so we proceed as in Steps~\ref{enum:nextevent} and~\ref{enum:update} for validity events.
The time complexity is $O(|\GGt|)$.

\medskip
\item[Cuddle event routine.] Handled (and detected)
as bending events, with the same time complexity.
\end{description}

\section{Analysis}\label{sec:analysis}
The correctness of our algorithm follows directly from Lemmas~\ref{lem:smoothrotation} and~\ref{lem:complete}. To analyze the time complexity, define $G\defeq\max_{0^{\circ}\leq\theta<180^{\circ}} |\GGt|$. We note that the number of validity events is~$2r$ and take $O(nG)$ time each, the number of domination and jumping events is $O(rG)$ and take $O(n)$ time each, and the number of passing, bending, and cuddling events is~$O(nr)$ (since they each associate a~vertex, either of the tour or of the polygon, with a reflex vertex of the polygon) and take $O(G)$ time each. Thus,
the complexity of our algorithm is~$O(nrG)$.

The storage complexity of our algorithm is $O(nG)$, since the algorithm maintains at each iteration of the main loop: the tour \Tt\ having linear size, the hourglass \HHt{\v}\ having linear size, the set \GGt\ having size $O(G)\in O(n)$, the change function $\len{\Tte}$ having size $O(G)\in O(n)$, the priority queue \QQ\ having size $O(G)\in O(n)$, and finally, at most $G$ visibility polygons, each of $O(n)$ size (dominating the storage usage).

It remains to prove that for a~fixed angle~$\theta$, we can quickly, in linear time, obtain a constant sized set \VV\ of polygon vertices so that \SP{\v,\v'}\ from \v\ to its image $\v'$ in $\HH{\v}(\theta)$ corresponds to a shortest watchman tour in \P\!, for some $\v\in\VV$. If $\P'(\theta)$ has two essential edges (corresponding to gates in \P) with the same color, we know from the proof of Lemma~\ref{lem:onlythreegates} that the highest reflex vertex along a~path between the gates in a~coordinate system where the gates are parallel to the $x$-axis must be touched by the tour, otherwise it is not the shortest. Since there are two paths between those gates, we obtain a set \VV\ of two vertices in this case. If $\P'(\theta)$ has one red and one blue essential edge, a shortest tour either touches one of the edge endpoints or, since the gates are parallel, it can be slide along the essential edge until it touches a reflex vertex in the polygon. We obtain such a reflex vertex by computing \SP{\p,\p'}\ for each essential edge endpoint \p\ and following \SP{\p,\p'}\ to the last vertex before it intersects the first gate in $\HH{\p}(\theta)$. We thus obtain a~set~\VV\ of at most eight vertices, four essential edge endpoints, and four vertices obtained from~$\HH{\p}(\theta)$.
We have shown the following theorem.
\begin{theorem}
The presented algorithm computes the minimum length shortest $($floating$)$ watchman tour under $\theta$-monotone visibility over all\/ $0^{\circ}\leq\theta<180^{\circ}$ in a simple polygon in $O(nrG)$ time and $O(nG)$ storage, where $n$ is the number of vertices, $r$ is the number of reflex vertices, and $G\leq r$ is the maximum number of gates of the polygon used at any time in the algorithm.
\end{theorem}

\section{Conclusions}\label{sec:conc}
Observe that our approach can also be used to obtain optimal tours for other parameters which are dependent on the rotation angle $\theta$, e.g., the longest of all shortest watchman tours, the one with smallest or largest area, etc. All we need is to adapt Lemma~\ref{lem:smoothrotation} for the specific problem.

An interesting extension of our problem is to minimize the longest out of multiple tours that together see a polygon under rotated monotone visibility. For standard visibility the problem is known to be NP-hard even for two tours~\cite{MitWyn:watchmen} and efficient constant factor approximation algorithms also exist for this case~\cite{NilPac-2watchman:conf}.


\section{Acknowledgements}
Bengt J.~Nilsson was supported by grant 2018-04001 from the Swedish Research Council.
%
David Orden was supported by project 734922-CONNECT from H2020-MSCA-RISE, by project MTM2017-83750-P from the Spanish Ministry of Science (AEI\-/FEDER, UE), and by project PID2019-104129GB-I00 from the Spanish Ministry of Science and Innovation.
%
Carlos Seara was supported by project 734922-CONNECT from  H2020-MSCA-RISE, and by projects MTM2015-63791-R MINE\-CO/FEDER and Gen.~Cat.~DGR~2017SGR1640.
%
Pawe\l{} \.Zyli\'nski was supported by grant 2015/17/B/ST6/01887 from the Polish National Science Centre.

\vspace*{-2ex}%


\begin{thebibliography}{10}

\bibitem{AleOrPaSeUr:capturingpoints}
 {\sc C.~Alegr\'ia-Galicia, D.~Orden, L.~Palios, C.~Seara, J.~Urrutia}.
\newblock Capturing points with a rotating polygon (and a 3D extension).
\newblock {\em Theory Comput.~Syst.} 63(3):543--566,~2019.

\bibitem{AleOrSeUr:convexhull}
{\sc C.~Alegr\'ia-Galicia, D.~Orden, C.~Seara, J.~Urrutia}. 
\newblock Efficient computation of minimum-area rectilinear convex hull under rotation and generalizations.
\newblock \texttt {arXiv:1710.10888v2},~2019.

\bibitem{AsaGhoShe:inhandbook}
{\sc T.~Asano, S.K.~Ghosh, T.~Shermer}.
\newblock Chapter~19. Visibility in the plane.
\newblock In
{\sc J.R.~Sack, J.~Urrutia}, editors.
\newblock {\em Handbook on Computational Geometry}.
\newblock Elsevier Science Publishers,~1999.

\bibitem{CarJonNil:journal}
{\sc S.~Carlsson, H.~Jonsson, B.J.~Nilsson}.
\newblock Finding the shortest watchman route in a simple polygon.
\newblock {\em Discrete Comp.~Geom.} 22:377--402,~1999.

\bibitem{Cha:triangulation}
{\sc B.~Chazelle}.
\newblock Triangulating a simple polygon in linear time.
\newblock {\em Discrete Comp.~Geom.} 6(3):485--524,~1991. 

\bibitem{ChiNta:rectwatchman}
{\sc W.~Chin, S.~Ntafos}.
\newblock Optimum watchman routes.
\newblock {\em Inf.~Proc.~Lett.} 28:39--44,~1988.

\bibitem{ChiNta:watchman}
{\sc W.~Chin, S.~Ntafos}.
\newblock Shortest watchman routes in simple polygons.
\newblock {\em Discrete Comp.~Geom.} 6(1):9--31,~1991.

\bibitem{DroEfrLubMit:watchman}
{\sc M.~Dror, A.~Efrat, A.~Lubiw, J.S.B.~Mitchell}.
\newblock Touring a sequence of polygons.
\newblock {\em 35th STOC'03}, pages 473--482,~2003.

\bibitem{DumTot:holewatchman}
{\sc A.~Dumitrescu, C.D.~T\'oth}.
\newblock Watchman tours for polygons with holes.
\newblock {\em Comput.~Geom.}
45(7):326--333,~2012.

\bibitem{G}
{\sc L.P.~Gewali}.
\newblock Recognizing $s$-star polygons.
\newblock {\em Patt.~Rec.} 28(7):1019--1032,~1995.

\bibitem{GuiHerLevShaTar:visandshortpath}
{\sc L.~Guibas, J.~Hershberger, D.~Leven, M.~Sharir, R.~Tarjan}.
\newblock Linear time algorithms for visibility and shortest path problems inside triangulated simple polygons.
\newblock {\em Algorithmica}. 2:209--233,~1987.

\bibitem{HamNil-SWR:conf}
{\sc M.~Hammar, B.J.~Nilsson}.
\newblock Concerning the time bounds of existing shortest watchman route algorithms.
\newblock {\em 11th FCT'97}, LNCS~1279, pages~210--221,~1997.

\bibitem{Her:catoptrica}
{\sc Heron of Alexandria}.
\newblock Catoptrica (On Reflection).~$\sim\!62$.

\bibitem{HerSur:rayshoot}
{\sc J.~Hershberger, S.~Suri}.
\newblock A pedestrian approach to ray shooting: shoot a ray, take a walk.
\newblock {\em Journal of Algorithms}. 18(3):403--431,~1995.

\bibitem{HBZ19}
{\sc G.~Huski{\'c}, S.~Buck, A.~Zell}.
\newblock {GeRoNa}: Generic robot navigation.
\newblock {\em J.~Intell.~Robot.~Syst.} 95(2):419--442,~2019.

\bibitem{IK}
{\sc C.~Icking, R.~Klein}.
\newblock Searching for the kernel of a polygon: a competitive strategy.
\newblock {\em 11th SoCG'95}, pages 258--266,~1995.

\bibitem{LeePre:shortpath}
{\sc D.T.~Lee, F.P.~Preparata}.
\newblock Euclidean shortest paths in the presence of rectilinear barriers.
\newblock {\em Networks}. 14:393--410,~1984.

\bibitem{Mit:approxwatchman}
{\sc J.S.B.~Mitchell}.
\newblock Approximating watchman routes.
\newblock {\em 24th SODA'13}, pages 844--855,~2013.

\bibitem{MitWyn:watchmen}
{\sc J.S.B.~Mitchell, E.L. Wynters}.
\newblock Watchman routes for multiple guards.
\newblock {\em 3rd CCCG'91}, pages~126--129, 1991.

\bibitem{NilPac-2watchman:conf}
{\sc B.J.~Nilsson, E.~Packer}.
\newblock An approximation algorithm for the two-watchman route in a simple polygon.
\newblock {\em 32nd EuroCG'16}, pages~111--114,~2016.

\bibitem{OrPaSeUrZy:widestcorridor}
{\sc D.~Orden, L.~Palios, C.~Seara, J.~Urrutia, P.~\.Zyli\'nski}. 
\newblock On the width of the monotone-visibility kernel of a simple polygon.
\newblock {\em 36th EuroCG'20}, pages~13:1--13:8,~2020.

\bibitem{EuroCG18}
{\sc D.~Orden, L.~Palios, C.~Seara, P.~\.Zyli\'nski}.
\newblock Generalized kernels of polygons under rotation.
\newblock {\em 34th EuroCG'18}, pages~74:1--74:6,~2018.

\bibitem{P}
{\sc L.~Palios}.
\newblock An output-sensitive algorithm for computing the $s$-kernel.
\newblock {\em 27th CCCG'15}, pages~199--204,~2015.

\bibitem{SRW}
{\sc S.~Schuierer, G.J.E.~Rawlins, D.~Wood}.
\newblock A generalization of staircase visibility.
\newblock {\em 3rd CCCG'91}, pages~96--99,~1991.

\bibitem{SW}
{\sc S.~Schuierer, D.~Wood}.
\newblock Generalized kernels of polygons with holes.
\newblock {\em 5th CCCG'93}, pages~222--227,~1993.

\bibitem{SW98}
{\sc S.~Schuierer, D.~Wood}.
\newblock Multiple-guards kernels of simple polygons.
\newblock {\em Theoretical Computer Science Center Research Report HKUST-TCSC-98-06},~1998.

\bibitem{Tan:watchman}
{\sc X.-H.~Tan}.
\newblock Fast computation of shortest watchman routes in simple polygons.
\newblock {\em Inf.~Process.~Lett.} 77(1):27--33,~2001.

\bibitem{TanHir:watchman}
{\sc X.-H.~Tan, T.~Hirata}.
\newblock Constructing shortest watchman routes by divide and conquer.
\newblock {\em 4th ISAAC'93}, LNCS~762, pages~68--77,~1993.

\bibitem{TanHirIna:watchman}
{\sc X.-H.~Tan, T.~Hirata, Y.~Inagaki}.
\newblock An incremental algorithm for constructing shortest watchman routes.
\newblock {\em Int.~J.~Comput.~Geom.~Appl.} 3(4):351--365,~1993.

\bibitem{TanHirIna:corrwatchman}
{\sc X.-H.~Tan, T.~Hirata, Y.~Inagaki}.
\newblock Corrigendum to ``An incremental Algorithm for constructing shortest
  watchman routes''.
\newblock {\em Int.~J.~Comput.~Geom.~Appl.} 9(3):319--324,~1999.

\end{thebibliography}
\end{document}